\documentclass[a4paper,final]{scrartcl}
\usepackage{amsmath}
\usepackage{amssymb}
\usepackage{amsthm}
\usepackage{xcolor}
\usepackage{enumerate}
\usepackage{enumitem}
\usepackage[USenglish]{babel}
\usepackage{ifthen}
\usepackage[notref,notcite,color]{showkeys}
\usepackage{booktabs}
\usepackage{multirow}
\usepackage{xspace}
\usepackage{url}
\usepackage{stmaryrd}
\usepackage{multicol}
\usepackage{authblk} 

\usepackage{tikz}
\usetikzlibrary{shapes,matrix,arrows}

\usetikzlibrary{decorations.pathmorphing}

\definecolor{darkgreen}{rgb}{0.0,0.7,0.0}
\newenvironment{tw}{\noindent\color{darkgreen} TW:}{}

\newenvironment{mk}{\noindent\color{blue} MK:} {}

\newcommand{\refthm}[1]{Theorem~\ref{#1}\xspace}

\newcommand{\reflem}[1]{Lemma~\ref{#1}\xspace}
\newcommand{\refprop}[1]{Proposition~\ref{#1}\xspace}

\newcommand{\set}[2]{\left\{#1\mathrel{\left|\vphantom{#1}\vphantom{#2}\right.}#2\right\}}
\newcommand{\oneset}[1]{\left\{\mathinner{#1}\right\}}

\let\iff=\undefined
\let\implies=\undefined

\newcommand{\iff}       {\mathrel{\Leftrightarrow}}
\newcommand{\implies}   {\text{$\;\Rightarrow\;$}}

\newcommand{\B}{\mathbb{B}}

\newcommand{\logicfont}[1]{\mathrm{#1}}

\newcommand{\FO}{\logicfont{FO}}

\newcommand{\Sigmatwo}{\mbox{\ensuremath{\Sigma_2}}}

\newcommand{\ltrue}      {\ensuremath{\mathord{\top}}\xspace}

\renewcommand{\lor}{\mathrel{\vee}}
\renewcommand{\land}{\mathrel{\wedge}}

\newcommand{\Synt}{\mathrm{Synt}}
\newcommand{\greenfont}[1] {\ensuremath{\mathcal{#1}}}

\newcommand{\greenR} {\greenfont{R}\xspace}
\newcommand{\Requiv} {\mathrel{\greenR}}

\let\alph=\undefined
\newcommand{\alph}{{\mathrm{alph}}}
\newcommand{\im}{{\mathrm{im}}}

\renewcommand{\phi}{\varphi}

\newcommand{\mediumSize}[1]{\fontsize{9pt}{12pt}\selectfont #1\normalsize}
\newcommand{\mediumFont}[1]{\normalfont\mediumSize{#1}}
\newcommand{\malcev}%
  {\mathop{\text{\normalsize{\raisebox{0.3mm}{\textcircled{\raisebox{0.1mm}{\mediumFont{m}}}}}}}}
\newcommand{\varietyFont}[1]{\mathrm{\mathbf{#1}}}

\newcommand{\Vthreehalf}{\varietyFont{V_{3/2}}}
\newcommand{\Vtwo}{\varietyFont{V_2}}

 \newtheorem{theorem}{Theorem}

 \newtheorem{proposition}[theorem]{Proposition}
 \newtheorem{lemma}[theorem]{Lemma}
 
 \theoremstyle{remark}
 \newtheorem{remark}[theorem]{Remark}
 \newtheorem{example}[theorem]{Example}
\setlist{itemsep=3pt,parsep=3pt,topsep=4pt}

\title{Level Two of the Quantifier Alternation Hierarchy over Infinite Words}
\begin{document}
\author{Manfred Kuf\-leitner}
\author{Tobias Walter}
\affil{FMI, Universit\"at Stuttgart, Germany,
	\{kufleitner,walter\}@fmi.uni-stuttgart.de}
\date{September 21, 2015\thanks{This work was supported by the German Research Foundation (DFG) under grants DI 435/5-2 and DI 435/6-1.}}

\maketitle

\begin{abstract}
The study of various decision problems for logic fragments has a long history in computer science. This paper is on the membership problem for a fragment of first-order logic over infinite words; the membership problem asks for a given language whether it is definable in some fixed fragment.
The alphabetic topology was introduced as part of an effective characterization of the fragment~$\Sigma_2$ over infinite words. Here, $\Sigma_2$ consists of the first-order formulas with two blocks of quantifiers, starting with an existential quantifier. Its Boolean closure is $\mathbb{B}\Sigma_2$. Our first main result is an effective characterization of the Boolean closure of the alphabetic topology, that is, given an $\omega$-regular language $L$, it is decidable whether $L$ is a Boolean combination of open sets in the alphabetic topology. This is then used for transferring Place and Zeitoun's recent decidability result for $\mathbb{B}\Sigma_2$ from finite to infinite words.
\end{abstract}

\section{Introduction}

Over finite words, the connection between finite monoids and regular languages is highly successful for studying logic fragments, see e.g.~\cite{dgk08ijfcs:short,str94:short}. Over infinite words, the algebraic approach uses infinite repetitions. Not every logic fragment can express whether some definable property $P$ occurs infinitely often. For instance, the usual approach for saying that $P$ occurs infinitely often is as follows: for every position $x$ there is a position $y>x$ satisfying $P(y)$. Similarly, $P$ occurs only finitely often if there is a position $x$ such that all positions $y>x$ satisfy $\neg P(y)$. Each of these formulas requires (at least) one additional change of quantifiers, which not all fragments can provide. It turns out that topology is a very useful tool for restricting the infinite behaviour of the algebraic approach accordingly, see e.g.~\cite{dk11tocs,kkl11stacs:short,pp04:short,wilke93stacs}. In particular, the combination of algebra and topology is convenient for the study of languages in $\Gamma^\infty$, the set of finite and infinite words over the alphabet~$\Gamma$.
In this paper, a \emph{regular language} is a regular subset of $\Gamma^\infty$.

\enlargethispage{-0.9\baselineskip}

Topological ideas have a long history in the study of $\omega$-regular languages. The Cantor topology is the most famous example in this context. 
We write $G$ for the Cantor-open sets and $F$ for the closed sets. 
The open sets in $G$ are the languages of the form $W\Gamma^\infty$ for $W \subseteq \Gamma^*$.
If $X$ is a class of languages, then $X_\delta$ consists of the countable intersections of languages in $X$ and $X_\sigma$ are the countable unions; moreover, we write $\mathbb{B}X$ for the Boolean closure of $X$. Since $F$ contains the complements of languages in $G$, we have $\mathbb{B}F = \mathbb{B}G$. The Borel hierarchy is defined by iterating the operations $X \mapsto X_\delta$ and $X \mapsto X_\sigma$. The Borel hierarchy over the Cantor topology has many appearances in the context of $\omega$-regular languages. For instance, an $\omega$-regular language is deterministic if and only if it is in $G_\delta$, see~\cite{Landweber69,tho90handbook:short}. By McNaughton's Theorem~\cite{mcnaughton66}, every $\omega$-regular language is in $\mathbb{B}(G_\delta) = \mathbb{B}(F_\sigma)$. The inclusion $\mathbb{B} G \subset G_\delta \cap F_\sigma$ is strict, but the $\omega$-regular languages in $\mathbb{B} G$ and $G_\delta \cap F_\sigma$ coincide~\cite{sw74eik:short}.

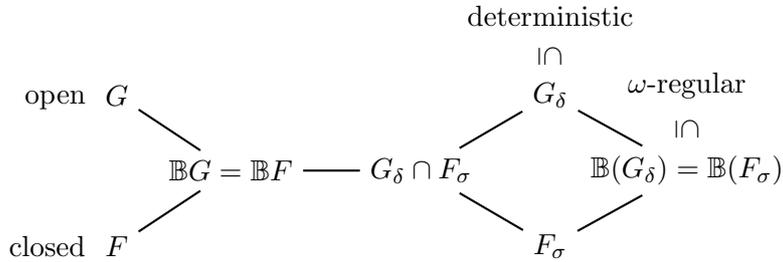
\vspace*{-\baselineskip}

\begin{center}
\begin{tikzpicture}
\draw (0,1) node (G) {$G$};
\draw (0,-1) node (F) {$F$};
\draw (1.5,0) node (BG) {$\mathbb{B}G = \mathbb{B}F$};
\draw (4,0) node (GdFs) {$G_\delta \cap F_\sigma$};
\draw (5.7,1) node (Gd) {$G_\delta$};
\draw (5.7,-1) node (Fs) {$F_\sigma$};
\draw (7.5,0) node (BGd) {$\mathbb{B}(G_\delta) = \mathbb{B}(F_\sigma)$};

\draw[thick] (G) -- (BG) -- (GdFs) -- (Gd) -- (BGd);
\draw[thick] (F) -- (BG);
\draw[thick] (GdFs) -- (Fs) -- (BGd);

\draw (G) ++(0,-0.07) node[left,outer sep=8pt] {open};
\draw (F) node[left,outer sep=8pt] {closed};
\draw (Gd) node[above,outer sep=11pt] (sseA) {};
\draw (sseA) node[rotate=-90] {$\subseteq$};
\draw (sseA) node[above,outer sep=8pt] {deterministic};
\draw (BGd) node[above,outer sep=12pt] (sseB) {};
\draw (sseB) node[rotate=-90] {$\subseteq$};
\draw (sseB) node[above,outer sep=8pt] {$\omega$-regular};

\end{tikzpicture}
\end{center}

Let $\FO^k$ be the fragment of first-order logic which uses (and reuses) at most $k$ variables. By $\Sigma_m$ we denote the formulas with $m$ quantifier blocks, starting with a block of existential quantifiers. Here, we assume that $x<y$ is the only binary predicate. We frequently identify a fragment with the languages definable therein. Let us consider $\FO^1$ as a toy example. With only one variable, we cannot make use of the binary predicate $x<y$. Therefore, in $\FO^1$ we can say nothing but which letters occur, that is, a language is definable in $\FO^1$ if and only if it is a Boolean combination of languages of the form $\Gamma^* a \Gamma^\infty$ for $a \in \Gamma$. Thus $\FO^1 \subseteq \mathbb{B}G$. It is an easy exercise to show that a regular language is in $\FO^1$ if and only if it is in $\mathbb{B}G$ and its syntactic monoid is both idempotent and commutative. The algebraic condition without the topology is too powerful since this would also include the language $\oneset{a,b}^*a^\omega$, which is not in $\FO^1$. For the fragment $\mathbb{B}\Sigma_1$, the same topology $\mathbb{B}G$ with a different algebraic condition works, cf.~\cite[Theorems~VI.3.7, VI.7.4 and~VIII.4.5]{pp04:short}.

In the fragment $\Sigma_2$, we can define the language $\oneset{a,b}^* a b^\infty$ which is not deterministic and hence not in $G_\delta$. Since the next level of the Borel hierarchy already contains all regular languages, another topology is required. For this purpose, Diekert and the first author introduced the \emph{alphabetic topology}~\cite{dk11tocs}: the open sets in this topology are arbitrary unions of languages of the form $uA^\infty$ for $u\in \Gamma^*$ and $A \subseteq \Gamma$. They showed that a regular language is definable in $\Sigma_2$ if and only if it satisfies some particular algebraic property and if it is open in the alphabetic topology. Therefore, the canonical ingredient for an effective characterization of $\mathbb{B}\Sigma_2$ is the Boolean closure of the open sets in the alphabetic topology. Our first main result shows that, for a given regular language $L$, it is decidable whether $L$ is a Boolean combination of open sets in the alphabetic topology. As a by-product, we see that every $\omega$-regular language which is a Boolean combination of arbitrary open sets in the alphabetic topology can be written as a Boolean combination of $\omega$-regular open sets. This resembles a similar result for the Cantor topology~\cite{sw74eik:short}.

A major breakthrough in the theory of regular languages over finite words is due to Place and Zeitoun~\cite{PlaceZeitoun14icalp}. They showed that, for a given regular language $L \subseteq \Gamma^*$, it is decidable whether~$L$ is definable in $\mathbb{B}\Sigma_2$. This solved a longstanding open problem, see e.g.~\cite[Section~8]{pin97handbook:short} for an overview. To date, no effective characterization of $\mathbb{B}\Sigma_3$ is known. Our second main result is to show that this decidability result transfers to languages in $\Gamma^\infty$. If $\Vtwo$ is the algebraic counterpart of $\mathbb{B}\Sigma_2$ over finite words, then we show that~$\Vtwo$ combined with the Boolean closure of the alphabetic topology yields a characterization of $\mathbb{B}\Sigma_2$ over $\Gamma^\infty$. Combining the decidability of $\Vtwo$ with our first main result, the latter characterization is effective. The proof that $\mathbb{B}\Sigma_2$ satisfies both the algebraic and the topological restrictions follows a rather straightforward approach. The main difficulty is to show the converse: every language satisfying both the algebraic and the topological conditions is definable in $\mathbb{B}\Sigma_2$.

\section{Preliminaries}
\subsection*{Words}
Let $\Gamma$ be a finite alphabet. By $\Gamma^*$ we denote the set of finite words over $\Gamma$; we write $1$ for the empty word.
The set of infinite words is $\Gamma^\omega$ and the set 
of finite and infinite words is $\Gamma^\infty = \Gamma^* \cup \Gamma^\omega$. By $u,v,w$ we denote finite 
words and by $\alpha,\beta,\gamma$ we denote words in $\Gamma^\infty$. 
In this paper a \emph{language} is a subset of $\Gamma^\infty$. 
Let $L\subseteq \Gamma^*$ and $K \subseteq \Gamma^\infty$. 
As usually $L^*$ is the union of powers of $L$ and 
$LK = \set{u\alpha}{u\in L,\alpha\in K} \subseteq \Gamma^\infty$ is the concatenation of $L$ and $K$. 
By $L^\omega$ we denote the set of words which are an infinite concatenation of words in $L$ and 
the infinite concatenation $uu\cdots$ of the word $u$ is written $u^\omega$. 
A word $u = a_1 \ldots a_n$ is a (scattered) subword of $v$ if $v \in \Gamma^* a_1 \Gamma^* \ldots a_n\Gamma^*$.
The \emph{alphabet} of a word is the set of all letters which appear in the word.
The \emph{imaginary alphabet} $\im(\alpha)$ of a word $\alpha \in \Gamma^\infty$ is the set of letters which 
appear infinitely often in $\alpha$. Let $A^\im = \set{\alpha \in \Gamma^\infty}{\im(\alpha) = A}$ 
be the set of words with imaginary alphabet $A$.
In the following we will restrict us to the study of regular languages. 
A language $L \subseteq \Gamma^*$ is regular if it is recognized by a (deterministic) finite automaton. 
A language $K \subseteq \Gamma^\omega$ is regular if it is recognized by a B\"uchi automaton. A language 
$L \subseteq \Gamma^\infty$ is regular if $L\cap \Gamma^*$ and $L\cap \Gamma^\omega$ are regular. This is 
equivalent to being recognized by an \emph{extended B\"uchi automaton}~\cite{dg08SIWT:short}.

\subsection*{First-Order logic}
We consider first order logic $\FO$ over $\Gamma^\infty$. Variables range over the position of the word. 
The atomic formulas in 
this logic are $\ltrue$ for true, $x<y$ to compare two positions $x$ and $y$ and $\lambda(x)=a$ which is true if 
the word has an $a$ at position $x$. 
One may combine those atomic formulas with the boolean connectives $\lnot$,$\land$ and 
$\lor$ and quantifiers $\forall$ and $\exists$. A \emph{sentence} $\varphi$ is a $\FO$ formula 
without free variables. We write $\alpha \models \varphi$ 
if $\alpha\in\Gamma^\infty$ satisfies the sentence $\varphi$. 
The language defined by $\varphi$ is $L(\varphi) = \set{\alpha\in\Gamma^\infty}{\alpha \models \varphi}$.
We will classify the formula of $\FO$ by counting the number of quantifier alternations, that is the number of alternations of $\exists$ and $\forall$. 
The fragment $\Sigma_i$ of $\FO$ contains all $\FO$-formula in prenex normal form with $i$ blocks 
of quantifiers $\exists$ or $\forall$, starting with a block of existential quantors. 
The fragment $\B\Sigma_i$ contains all Boolean combinations of formulas in $\Sigma_i$.  
We are particularly interested in the fragment $\Sigmatwo$ and the Boolean combinations of formulas in $\Sigmatwo$. 
A language $L$ is definable in a fragment $\mathcal F$ (e.g. $\mathcal F$ is $\Sigmatwo$ or $\B\Sigmatwo$) if there exists 
a formula $\varphi \in \mathcal F$ such that $L = L(\varphi)$, i.e., if $L$ is definable by some $\varphi \in \mathcal F$. 
The classes of languages defined by $\Sigma_i$ and $\B\Sigma_i$ form a hierarchy, the quantifier alternation hierarchy. 
This hierarchy is strict, i.e., $\Sigma_i \subsetneq \B\Sigma_i \subsetneq \Sigma_{i+1}$ holds for all $i$, cf.~\cite{bk78jcss:short,tho82:short}.

\subsection*{Monomials}
A \emph{monomial} is a language of the form $A_0^* a_1 A_1^* a_2 \cdots A_{n-1}^*
a_n A_n^\infty$ for $n \geq 0$, $a_i \in \Gamma$ and $A_i \subseteq \Gamma$. 
The number $n$ is called the \emph{degree}. In particular, $A_0^\infty$ is a monomial of degree $0$.
A monomial is called $k$-monomial if it has degree at most $k$.
In \cite{dk11tocs} it is shown that a language $L\subseteq \Gamma^\infty$ is in $\Sigma_2$ if and only if 
it is a finite union of monomials. We are interested in $\B\Sigmatwo$ and thus in finite Boolean
combination of monomials $A_0^* a_1 A_1^* a_2 \cdots A_{n-1}^* a_n A_n^\infty$. 
For this, let $\equiv^\infty_k$ be the equivalence relation on $\Gamma^\infty$ such that 
$\alpha \equiv^\infty_k \beta$ if $\alpha$ and $\beta$ are contained in exactly the same $k$-monomials. 
Thus, $\equiv_k^\infty$-classes are Boolean combinations of monomials and every language in $\B\Sigmatwo$ is a union of $\equiv_k^\infty$-classes for some $k$. 
Further, since there are only finitely many monomials of degree $k$, there are only finitely many $\equiv_k^\infty$-classes. 
The equivalence class of some word $\alpha$ in $\equiv_k^\infty$ is denoted by $[\alpha]_k^\infty$.
Note, that such a characterization of $\B\Sigmatwo$ in terms of monomials does not yield a decidable characterization. 

Our characterization of languages $L\subseteq \Gamma^\infty$ in $\B\Sigmatwo$ is based on the characterization of languages in $\B\Sigmatwo$ over finite words. 
For this, we also introduce monomials over $\Gamma^*$. A \emph{monomial} over $\Gamma^*$ is a language of the form 
$A_0^* a_1 A_1^* a_2 \cdots A_{n-1}^* a_n A_n^*$ for $n \geq 1$, $a_i \in \Gamma$ and $A_i \subseteq \Gamma$. 
The degree is defined as above. 
Let $\equiv_k$ be the congruence on $\Gamma^*$ which is defined by $u \equiv_k$ if and only if 
$u$ and $v$ are contained in the same monomials over $\Gamma^*$. The equivalence classes are noted by $[u]_k$.
Again, a language $L \subseteq \Gamma^*$ is in $\B\Sigmatwo$ over $\Gamma^*$ if and only if it is a union of $\equiv_k$-classes for some $k$, i.e., if $L = \cup_{u\in L} [u]_k$.

\subsection*{Algebra}

In this paper all monoids are either finite or free. Finite monoids are a common way for defining regular languages. A monoid element $e$ is \emph{idempotent} if $e^2=e$. Every element $x$ of a finite monoid admits a unique idempotent $x^i$ for some integer $i \geq 1$. 
An \emph{ordered monoid} $(M,\leq)$ is a monoid equipped with a partial order which is compatible with the monoid multiplication, i.e., $s\leq t$ and $s' \leq t'$ implies $ss' \leq tt'$. Every monoid can be ordered by using the identity as partial order. For a homomorphism $h : (N,\leq) \to (M,\leq)$ between ordered monoids we require $s \leq t \implies h(s) \leq h(t)$ for all $s,t \in N$. A \emph{divisor} is the homomorphic image of a submonoid.

A class of monoids which is closed under division and finite direct products is a \emph{pseudovariety}. 
Eilenberg showed a correspondence between certain classes of languages (of finite words) and pseudovarieties \cite{eil76:short}. 
A homomorphism $h : (N,\leq) \to (M,\leq)$ between two ordered monoids must hold $s \leq t \implies h(s) \leq h(t)$ for $s,t \in N$. 
A pseudovariety of ordered monoids is defined defined the same way as with unordered monoids, using the homomorphisms of ordered monoids. 
The Eilenberg correspondence then also holds for ordered monoids \cite{pin95:short}.
Let $\Vthreehalf$ be the pseudovariety of ordered monoids which corresponds to $\Sigmatwo$ and 
$\Vtwo$ be the pseudovariety of monoids which corresponds to languages in $\B\Sigmatwo$. 
Since $\Sigmatwo \subseteq \B\Sigmatwo$, we obtain $\Vthreehalf \subseteq \Vtwo$ when ignoring the order.
The connection between monoids and languages is given by the notion of \emph{recognizability}.
A language $L \subseteq \Gamma^*$ is \emph{recognized} by an ordered monoid $(M, \leq)$ 
if there is a monoid homomorphism $h : \Gamma^* \to M$ such that $L = \cup\set{h^{-1}(t)}{s \leq t \text{ for some } s\in h(L)}$. If $M$ is not ordered, then this means that $L$ is an arbitrary union of languages of the form $h^{-1}(t)$.

For $\omega$-languages $L \subseteq \Gamma^\infty$ the notion of recognizability is slightly more technical. For simplicity, we only consider recognition by non-ordered monoids. Let $h : \Gamma^* \to M$ be a monoid homomorphism. If the homomorphism $h$ is understood, we write $[s]$ for the language $h^{-1}(s)$. We call $(s,e)\in M\times M$ a \emph{linked pair} if $e^2 = e$ and $se = s$. By Ramsey's Theorem~\cite{ram30:short} for every word $\alpha \in \Gamma^\infty$ there exists a linked pair $(s,e)$ such that $\alpha \in [s][e]^\omega$.
A language $L \subseteq \Gamma^\infty$ is recognized by $h$ if 
\[
L = \bigcup \set{[s][e]^\omega}{(s,e) \text{ is a linked pair with } [s][e]^\omega \cap L \neq \emptyset}.
\]
Since $1^\omega = 1$, the language $[1]^\omega$ also contains finite words. We thus obtain recognizability of languages of finite words as a special case. A language $L \subseteq \Gamma^\infty$ is \emph{regular} if it is recognized by (a homomorphism to) a finite monoid.


Next, we define syntactic homomorphisms and syntactic monoids; as we will see, these are the minimal recognizers of a regular language. 
Let $L \subseteq \Gamma^\infty$ be a regular language. 
The syntactic monoid of $L$ is defined as the quotient $\Synt(L) = \Gamma^*/\!\approx_L$ where $u \approx_L v$ holds if and only if for all $x,y,z \in \Gamma^*$ we have both $xuyz^\omega \in L \iff xvyz^\omega$ and $x(uy)^\omega \in L \iff x(vy)^\omega \in L$. 
The syntactic monoid can be ordered by the quasiorder $\preceq_L$ defined by $u \preceq_L v$ if for all $x,y,z \in \Gamma^*$ we have $xuyz^\omega \in L \implies xvyz^\omega$ and $x(uy)^\omega \in L \implies x(vy)^\omega \in L$. 
One can effectively compute the syntactic homomorphism of $L$.
The syntactic monoid $\Synt(L)$ satisfies the property that $L$ is regular if and only if $\Synt(L)$ is finite and the canonical 
homomorphism $h_L : \Gamma^* \to \Synt(L)$ recognizes $L$, see~e.g.~\cite{pp04:short,tho90handbook:short}.
Every pseudovariety is generated by its syntactic monoids \cite{eil76:short}, i.e., every monoid in a given pseudovariety is a divisor of a direct product of syntactic monoids.
The importance of the syntactic monoid of some language $L\subseteq \Gamma^\infty$ is that it is the smallest monoid recognizing $L$:

\begin{lemma}\label{lem:syntminimal}
	Let $L\subseteq \Gamma^\infty$ be a language which is recognized by a homomorphism $h : \Gamma^* \to (M,\leq)$. 
	Then, $(\Synt(L), \preceq_L)$ is a divisor of $(M,\leq)$.
\end{lemma}
\begin{proof}
	We assume that $h$ is surjective and show that $\Synt(L)$ is a quotient of $M$. 
	If $h$ is not surjective, we can therefore conclude that $\Synt(L)$ is a divisor of $M$.
	We show that $h(u) \leq h(v) \implies u \preceq_L v$. Let $u,v$ be words with $h(u) \leq h(v)$ and denote $h^{-1}(h(w)) = [h(w)]$ for words $w$. 
	Assume $xuyz^\omega \in L$, then there exists an index $i$ such that $(h(xuyz^i), h(z)^\omega)$ is a linked pair. 
	Thus, $[h(xuyz^i)][h(z)]^\omega \subseteq L$ and by $h(u) \leq h(v)$ also $[h(xvyz^i)][h(z)]^\omega \subseteq L$. This implies $xvyz^\omega \in L$. 
	The proof that $x(uy)^\omega \in L \implies x(vy)^\omega \in L$ is similar. Thus, $u \preceq_L v$ holds which shows the claim.
\end{proof}

We stated the lemma for ordered monoids also for languages containing infinite words, but in the ordered setting it will be applied only for finite words.

\section{Alphabetic Topology}\label{sec:top}
The topological component is crucial for our approach. 
As mentioned in the introduction, combining algebraic and topological conditions is a successful approach for characterizations of language 
classes over $\Gamma^\infty$. 
A topology on a set $X$ is given by a family of subsets of $X$ (called open) which are closed under finite intersections and arbitrary unions. 
We define the \emph{alphabetic topology} over $\Gamma^\infty$ by its basis $\set{uA^\infty}{u\in \Gamma^*, A\subseteq \Gamma}$. 
Hence, an open set is described as $\bigcup_A W_A A^\infty$ with $W_A\subseteq \Gamma^*$. 
The alphabetic topology has been introduced in~\cite{dk11tocs}, where it is used as a part of the characterization of $\Sigmatwo$ over $\Gamma^\infty$.
\begin{theorem}[\cite{dk11tocs}]
	Let $L \subseteq \Gamma^\infty$ be a regular language. Then $L \in \Sigmatwo$ if and only if $\Synt(L) \in \Vthreehalf$ and $L$ is open in the alphabetic topology. 
\end{theorem}
The alphabetic topology has by itself been the subject of further study \cite{SchwarzStaiger10}.
We are particularly interested in Boolean combinations of open sets. 
An effective characterization of a language $L$ being a Boolean combination of open sets in the alphabetic topology 
is given in the proposition below. 
\begin{theorem}\label{thm:boolalphopen}
Let $L\subseteq \Gamma^\infty$ be a regular language which is recognized by $h : \Gamma^* \to M$. 
Then the following are equivalent:
\begin{enumerate}
\item $L$ is a Boolean combination of open sets in the alphabetic topology where each open set is regular.\label{boolalphopen:aaaa}
\item $L$ is a Boolean combination of open sets in the alphabetic topology.\label{boolalphopen:bbbb}
\item For all linked pairs $(s,e),(t,f)$ it holds that if there exists an alphabet $C$ and words $\hat e, \hat f$ with 
$h(\hat e) = e, h(\hat f) = f$, $\alph(\hat e) = \alph(\hat f) = C$ and $s\cdot h(C^*) = t\cdot h(C^*)$, 
then $[s][e]^\omega \subseteq L \iff [t][f]^\omega \subseteq L$.\label{boolalphopen:cccc}
\end{enumerate}
\end{theorem}
\begin{proof}
``\ref{boolalphopen:aaaa} $\Rightarrow$ \ref{boolalphopen:bbbb}'': This is trivial.

%
%
``\ref{boolalphopen:bbbb} $\Rightarrow$ \ref{boolalphopen:cccc}'': 
Let $L$ be a Boolean combination of strict alphabetic open sets. 
We may assume 
\[
L = \bigcup_{i=1}^n \left((P_iA_i^\infty) \setminus \left( \bigcup_{j=1}^{m_i} Q_{i,j}B_{i,j}^\infty\right)   \right)
\]
for some $P_i, Q_{i,j} \subseteq \Gamma^*$ and alphabets $A_i, B_{i,j} \subseteq \Gamma$.
Assume $[s][e]^\omega \subseteq L$, but $[t][f]^\omega \not\subseteq L$. It suffices to show that $[t][f]^\omega \cap L$ is nonempty. 
Let $u\hat e^\omega\in [s][e]^\omega \subseteq L$ for some $u\in [s], \hat e\in[e]$ with $\alph(\hat e) = C$. 
We also choose some words $\hat f,x,y\in C^*$ such that $h(\hat f) = f$, $s\cdot h(x) = t$, $t\cdot h(y) = s$ and $\alph(\hat f) = C$.

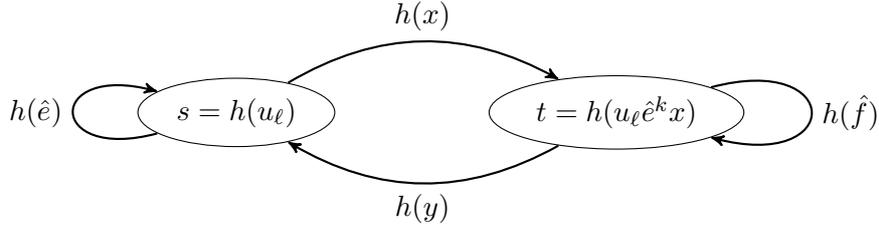
\begin{figure}
	\begin{center}		
		\begin{tikzpicture}[>=stealth',node distance = 5cm]
		\node[draw,ellipse] (s) {$s=h(u_\ell)$};
		\node[draw,ellipse,right of=s] (t) {$t=h(u_\ell \hat e^k x)$}; 
		\path[thick,->] (s) edge[bend left] node[above] {$h(x)$} (t);
		\path[thick,->] (t) edge[bend left] node[below] {$h(y)$} (s);
		\path[thick,->] (s) edge[loop left, looseness=7] node[left] {$h(\hat e)$} (s);
		\path[thick,->] (t) edge[loop right, looseness=7] node[right] {$h(\hat f)$} (t);
		\end{tikzpicture}
	\end{center}
	\caption{Part of the right Cayley graph of $M$ in the proof of ``\ref{boolalphopen:bbbb} $\Rightarrow$ \ref{boolalphopen:cccc}''.}
\end{figure}

The idea is to find an increasing sequence of words $u_\ell\in [s]$ and sets $I_\ell \subseteq \oneset{1,\ldots, n}$ such that 
$u_\ell C^\infty  \cap \left(P_iA_i^\infty \setminus \left( \bigcup_{j=1}^{m_i} Q_{i,j}B_{i,j}^\infty\right)\right) = \emptyset$ for all 
$i \in I_\ell$.
We can set $u_0 = u$ and $I_0 = \emptyset$. Consider the word $u_\ell \hat e^\omega \in L$. 
There exists an index $i\in \oneset{1,\ldots,n}\setminus I_\ell$ such that 
$u_\ell \hat e^\omega \in P_iA_i^\infty \setminus \left( \bigcup_{j=1}^{m_i} Q_{i,j}B_{i,j}^\infty\right)$. 
Choose $k$ big enough, such that in the decomposition $u_\ell \hat e^k \hat e^\omega$ the part $u_\ell \hat e^k$ 
overlaps into the $A_i^\infty$ part.
Since $C = \alph(\hat e) \subseteq A_i$, we also have $ \beta_\ell = u_\ell \hat e^k x \hat f^\omega \in P_iA_i^\infty \cap A_i^\im$. 
By construction we have $\beta_\ell \in [t][f]^\omega$ and therefore, 
assuming $[t][f]^\omega \cap L = \emptyset$, there exists an index $j$ such that 
$\beta_\ell \in Q_{i,j}B_{i,j}^\infty$.  Analogous, there exists a $k'$ such that 
$u_\ell \hat e^k x \hat f^{k'} y C^\infty \subseteq Q_{i,j}B_{i,j}^\infty$. 
Hence we can choose $u_{\ell+1} = u_\ell \hat e^k x \hat f^{k'} y$ and $I_{\ell+1} = I_\ell \cup \oneset{i}$.

Since $u_\ell[e]^\omega \subseteq L\cap u_\ell C^\infty$, this construction has to fail at an index $\ell<n$. 
Therefore, the assumption is not 
justified and we have $[t][f]^\omega \cap L \neq \emptyset$, proving the claim.

``\ref{boolalphopen:cccc} $\Rightarrow$ \ref{boolalphopen:aaaa}'': 
Let $\alpha\in [s][e]^\omega \subseteq L$ for a linked pair $(s,e)$. 
Let $C = \im(\alpha)$. By $\alpha\in [s][e]^\omega$ and the definition of $C$  there exists an $\hat e \in C^*$ with $\alph(\hat e) = C$ and $h(\hat e) = e$.
Define \[
L' := L(s,C) := [s]C^\infty \setminus \left( \bigcup_{D \subsetneq C} \Gamma^*D^{\infty} 
                              \cup \bigcup_{s\not\in t\cdot h(C^*)}[t]C^\infty\right).
\]
We have $\alpha \in L'$ and $L'$ is a Boolean combination of open sets in the alphabetic topology whereas each open set is regular. 
Since there are only finitely many sets of the type $L(s,C)$, 
it suffices to show $L' \subseteq L$. For $C = \emptyset$ we have $L' = [s]$ and hence $L' \subseteq L$. 
Thus, we may assume $C \neq \emptyset$. Let $\beta\in L'$ be an arbitrary element and let $\beta \in [t][f]^\omega$ for a linked pair $(t,f)$. 
Since $\beta$ is in $L'$, it admits a decomposition $\beta = \tilde v \tilde \beta$ with $\tilde v \in [s]$ and 
$\tilde \beta \in C^\omega$. Also, by $\beta \in [t][f]^\omega$, one gets $\beta = v\beta'$ with $v\in [t], \beta' \in [f]^\omega$.
Using $tf = t$ and $C \neq \emptyset$, we may assume that $|v| \geq |\tilde v|$, 
which implies $\beta' \in C^\omega$. Hence we have $t \in s\cdot h(C^*)$. 
By construction $\beta \not\in \bigcup_{s\not\in t\cdot h(C^*)}[t]C^\infty$ and therefore $s \in t\cdot h(C^*)$. It follows 
$s\cdot h(C^*) = t\cdot h(C^*)$. Since $\beta \not\in  \bigcup_{D \subsetneq C} \Gamma^*D^{\infty}$, we also have 
$\alph(\beta') = C$.
Using \ref{boolalphopen:cccc} it follows $\beta \in L$.
\end{proof}

The alphabetic topology above is a refinement of the well-known Cantor topology. 
The Cantor topology is given by the basis $u\Gamma^\infty$ for $u\in \Gamma^*$. A regular language $L$ is a Boolean combination of
 open sets in the Cantor topology 
if and only if $[s][e]^\omega \subseteq L \Leftrightarrow [t][f]^\omega \subseteq L$ for all linked pairs $(s,e)$ and $(t,f)$ of the 
syntactical monoid of $L$ with $s \Requiv t$, c.f. \cite{dk11tocs,pp04:short,tho90handbook:short}. \refthm{thm:boolalphopen} is a similar 
result, but one had to consider the alphabetic information of the linked pairs. Hence, one does not have $s \Requiv t$ as condition, but rather 
$\Requiv$-equivalence within a certain alphabet $C$.

\begin{remark}
	The \emph{strict alphabetic topology} over $\Gamma^\infty$, which is introduced in \cite{dk11tocs}, is given by the basis $\set{uA^\infty \cap A^\im}{u\in \Gamma^*, A\subseteq \Gamma}$ and the open sets are of the form $\bigcup_A W_A A^\infty \cap A^\im$ with $W_A\subseteq \Gamma^*$. 
	Reusing the proof of \refthm{thm:boolalphopen} it turns out, that it is equivalent to be a Boolean combination of open sets in the alphabetic topology and in the strictly alphabetic topology. 
	Since $uA^\infty = \bigcup_{B\subseteq A} uA^*B^\infty \cap B^\im$, every open set in the alphabetic topology is also open in the strict alphabetic topology. 
	Further, one can adapt the proof of ``\ref{boolalphopen:bbbb} $\Rightarrow$ \ref{boolalphopen:cccc}'' of  \refthm{thm:boolalphopen} to show 
	that if $L$ is a Boolean combination of open sets in the strict alphabetic topology, 
	then item \ref{boolalphopen:cccc} of  \refthm{thm:boolalphopen} holds. 
\end{remark}

\section{The fragment $\B\Sigmatwo$}
Place and Zeitoun have shown that $\B\Sigmatwo$ is decidable over finite words. In particular, they have shown that given the syntactic homomorphism of a language $L$, it is decidable if $L \in \B\Sigmatwo$. Let $\Vtwo$ be the pseudovariety of monoids which corresponds to the 
language variety of all languages contained in $\B\Sigmatwo$. Since every pseudovariety is generated by its syntactic monoids, 
the result of Place and Zeitoun can be stated as follows:
\begin{theorem}[\cite{PlaceZeitoun14icalp}]\label{thm:pzc}
	The pseudovariety $\Vtwo$ corresponding to the $\B\Sigmatwo$-definable languages in $\Gamma^*$ is decidable. 
\end{theorem}

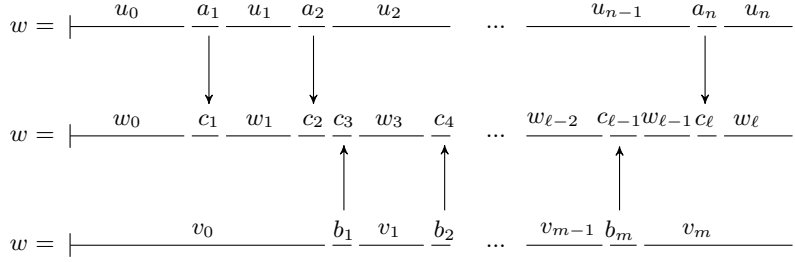
\begin{figure}
	\begin{center}
		\begin{tikzpicture}[>=stealth',scale=1,transform shape,font=\footnotesize]
		\draw[anchor=center] (1.25,1.2)  node{$u_0$} (2.325,1.2) node (a1) {$a_1$}  (3,1.2) node {$u_1$} (3.7,1.2) node (a2) {$a_2$} (4.7,1.2) node {$u_2$} (7.7,1.2) node {$u_{n-1}$} (8.85,1.2) node (an) {$a_n$} (9.55,1.2) node {$u_n$};
		\draw[anchor=center] (0,1)  node{$w=$}  (0.5,1)  node{$|$} -- (2,1) (2.1,1) -- (2.45,1) (2.55,1) -- (3.4,1) (3.5,1) -- (3.85,1) (3.95,1) -- (5.5,1) (5.6,1) (6.1,1) node{...} (6.5,1) -- (8.65,1) (8.75,1) -- (9,1) (9.1,1) -- (10,1) node{} (10.2,1);

		\draw[anchor=center] (1.25,-0.25)  node{$w_0$} (2.325,-0.25) node (c1) {$c_1$} (3,-0.25) node {$w_1$} (3.7,-0.25) node (c2) {$c_2$} (4.1,-0.25) node (c3) {$c_3$} (4.7,-0.25) node {$w_3$}  (5.425,-0.25) node (c4) {$c_4$} (6.85,-0.25) node {$w_{\ell-2}$} (7.725,-0.25) node (cl1) {$c_{\ell-1}$}  (8.35,-0.25) node {$w_{\ell-1}$}(8.85,-0.25) node (cl) {$c_\ell$} (9.4,-0.25) node {$w_\ell$};
		\draw[anchor=center] (0,-0.45)  node{$w=$}  (0.5,-0.45)  node{$|$} -- (2,-0.45) (2.1,-0.45) -- (2.45,-0.45) (2.55,-0.45) -- (3.4,-0.45) (3.5,-0.45) -- (3.85,-0.45) (3.95,-0.45) -- (4.2,-0.45) (4.3,-0.45) -- (5.15,-0.45) (5.25,-0.45) -- (5.5,-0.45) (5.6,-0.45) (6.1,-0.45) node{...} (6.5,-0.45) -- (7.5,-0.45) (7.6,-0.45) -- (7.95,-0.45) (8.05,-0.45) -- (8.65,-0.45) (8.75,-0.45) -- (9,-0.45) (9.1,-0.45) -- (10,-0.45) node{} (10.2,-0.45);
		
		\draw[anchor=center] (2.25,-1.7)  node{$v_0$} (4.1,-1.7) node (b1) {$b_1$} (4.7,-1.7) node {$v_1$}  (5.425,-1.7) node (b2) {$b_2$} (7.05,-1.7) node {$v_{m-1}$} (7.725,-1.7) node (bm) {$b_m$} (8.75,-1.7) node {$v_m$};
		\draw[anchor=center] (0,-1.9)  node{$w=$}  (0.5,-1.9)  node{$|$} -- (3.85,-1.9) (3.95,-1.9) -- (4.2,-1.9) (4.3,-1.9) -- (5.15,-1.9) (5.25,-1.9) -- (5.5,-1.9) (5.6,-1.9) (6.1,-1.9) node{...} (6.5,-1.9) -- (7.5,-1.9) (7.6,-1.9) -- (7.95,-1.9) (8.05,-1.9) -- (10,-1.9) node{} (10.2,-1.9);
		
		\draw[shorten <=0.1cm,->] (a1) to (c1);
		\draw[shorten <=0.1cm,->] (a2) to (c2);		
		\draw[shorten >=0.1cm,->] (b1) to (c3);
		\draw[shorten >=0.1cm,->] (b2) to (c4);
		\draw[shorten >=0.1cm,->] (bm) to (cl1);		
		\draw[shorten <=0.1cm,->] (an) to (cl);
		\end{tikzpicture}
	\end{center}
	\caption{Different factorizations in the proof of \reflem{lem:onemonomial}. In the situation of the figure it holds $C_0 = A_0 \cap B_0$, $C_1 = A_1 \cap B_0$, $C_2 = \emptyset$, $C_3 = A_2 \cap B_1$, $C_{\ell-2} = A_{n-1} \cap B_{m-1}$, $C_{\ell-1} = A_{n-1} \cap B_m$ and $C_\ell = A_n \cap B_m$.}
\end{figure}
The main part of the proof will be \refprop{prp:algebra2monomial}. The following lemma will be an auxiliary result for \refprop{prp:algebra2monomial}.
\begin{lemma}\label{lem:onemonomial}
	There exists a number $l$ such that for every set $\oneset{M_1,\ldots, M_d}$ of $k$-monomials over $\Gamma^*$ and every $w$ with $w \in M_i$ for all $i\in \oneset{1,\ldots,n}$, there exists a $l$-monomial $N$ over $\Gamma^*$ with $w \in N$ and $N \subseteq \cap M_i$.
\end{lemma}
\begin{proof}
	As one can iterate the statement, it suffices to show it for $d=2$. 
	Let $M_1 = A_0^* a_1 A_1^* a_2 \cdots A_{n-1}^*a_n A_n^*$ and 
	$M_2 = B_0^* b_1 B_1^* b_2 \cdots B_{m-1}^*b_m B_m^*$ be two monomials.
	Since $w \in M_1$ and $w \in M_2$, it admits factorizations 
	$w = u_0 a_1 u_1 a_2 \cdots u_{n-1} a_n u_n$ and $w = v_0 b_1 v_1 b_2 \cdots v_{m-1} b_m v_m$ such that 
	$u_i \in A_i^*$ and $v_i \in B_i^*$. 
	The factorizations mark the positions of the $a_i$s and the $b_j$s and pose an alphabetic conditions for the factors inbetween. 
	Thus, there exists a factorization $w = w_0 c_1 w_1 c_2 \cdots w_{\ell-1} c_\ell w_\ell$, 
	such that the positions of $c_i$ are exactly those, that are marked by $a_i$ or $b_j$, i.e., $c_i = a_j$ or $c_i = b_j$ for some $j$. 
	The words $w_i$ are over some alphabet $C_i$ such that $C_i = A_j \cap B_k$ for some $j$ and $k$
	induced by the factorizations.
	In the case of consecutive marked positions, one can set $C_i = \emptyset$. 
	Thus, we obtain a monomial $N = C_0^* c_1 C_1^* c_2 \cdots c_{\ell-1} C_{\ell-1}^* c_\ell C_\ell^*$ with $C_\ell = A_n \cap B_m$. 
	By construction $N \subseteq M_1$, $N \subseteq M_2$ and $w \in N$ holds.
	Since there are only finitely many monomials of degree $k$, the size of the number $l$ is bounded.
\end{proof}
An analysis of the proof of \reflem{lem:onemonomial} yields that the bound $l \leq n_k\cdot k$ holds, 
where $n_k$ is the number of distinct $k$-monomials over $\Gamma^*$. 
Next, we will show that a language which is in $\Vtwo$ and is a Boolean combination of alphabetic open sets is a finite Boolean combination of monomials. 
One ingredrient of the proof will be that by \reflem{lem:onemonomial}, we are able to compress the information of a set of $k$-monomials which contain a fixed word into the information that a single $l$-monomial contains that fixed word.

\begin{proposition}\label{prp:algebra2monomial}
	Let $L\subseteq \Gamma^\infty$ be a Boolean combination of alphabetic open sets such that $\Synt(L) \in \Vtwo$. Then $L$ is a finite Boolean combination of monomials.
\end{proposition}
\begin{proof}
	Let $h : \Gamma^* \to \Synt(L)$ be the syntactic homomorphism of $L$ and consider the languages $h^{-1}(p)$ for $p \in \Synt(L)$. 
	By \refthm{thm:pzc} we obtain $h^{-1}(p) \in \B\Sigmatwo$.
	Thus, there exists a number $k$ such that for every $p\in M$ the language $h^{-1}(p)$ is saturated by $\equiv_k$, i.e., $u \equiv_k v \implies h(u) = h(v)$. 
	By \reflem{lem:onemonomial} there exists a number $\ell$ such that for every set $\oneset{M_1,\ldots, M_n}$ of $k$-monomials and every $w$ with $w \in M_i$ for all $i\in \oneset{1,\ldots,n}$, there exists a $\ell$-monomial $N$ with $w \in N \subseteq \cap_{i=1}^n M_i$.
	Let $\alpha \equiv_\ell^\infty \beta$ and $\alpha \in L$. 
	We show $\beta \in L$ which implies $L = \cup_{\alpha \in L} [\alpha]_\ell^\infty$ and thus that $L$ is a finite Boolean combination of $\ell$-monomials.
	By observing membership in $\Gamma^*C^\infty$, it is clear that $\im(\alpha) = \im(\beta) =: C$.

	Let $u' \leq \alpha$ and $v' \leq \beta$ be prefixes such that for all every $\ell$-monomial $N = N' \cdot C^\infty$ with $\alpha, \beta \in N$ we have that some prefix of $u', v'$ is in $N'$. Further, let $u,v$ be the shortest prefixes of $\alpha,\beta$ such that $u' \leq u$, $v'\leq v$ and for $C = \oneset{c_1,\ldots, c_m}$ the word $(c_1c_2\cdots c_m)^k$ is a subword of $u''$ and $v''$ with $u = u' u''$ and $v = v' v''$, i.e., we extend the words $u'$ and $v'$ such that the full imaginary alphabet appears often enough. 
	Let $\alpha = u \alpha'$ and $\beta = v \beta'$.
	We use \refthm{thm:boolalphopen} and show that for $s = h(u)$ and $t = h(v)$ we have $s\cdot h(C^*) = t \cdot h(C^*)$, which implies $\beta \in L$.
	By symmetry, it suffices to show $t \in sh(C^*)$. Consider the set of $k$-monomials $N_i = N'_i C^\infty$ which hold at $u$, i.e., 
	such that $u\in N'_i$ and $\alpha' \in C^\infty$. 
	By the choice of $\ell$, there exists an $\ell$-monomial $N'$ such that $u \in N'$ and $N' \subseteq \cap_i N'_i$.
	Since $u \in N'$, we obtain $\alpha \in N := N' C^\infty$ and by $\alpha \equiv_\ell^\infty \beta$ the membership $\beta \in N$ holds. 
	By construction of $v$, there exists a prefix $\hat v \leq v' \leq v$ such that $\hat v \in N'$ and $\hat \beta \in C^\infty$ with $\hat \beta$ being defined by $\beta = \hat v \hat \beta$.
	Let $v = \hat v x$, then $x \in C^*$. We show that $ux \equiv_k v$.

		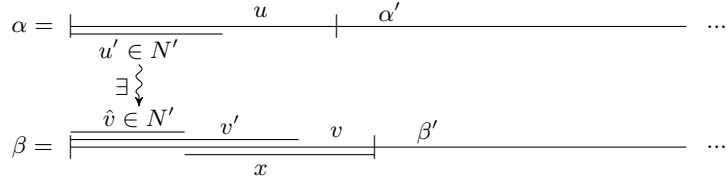
\begin{figure}
			\begin{center}
				\begin{tikzpicture}[>=stealth',scale=1,transform shape,font=\footnotesize]
				\draw[anchor=center] (1.4,0.7) node {$u'\in N'$} (3,1.2) node (a1) {$u$} (4.7,1.2) node {$\alpha'$};
				\draw[anchor=center] (0,1)  node{$\alpha=$}  (0.5,1)  node{$|$} -- (2,1) -- (4,1) node{$|$} -- (5.5,1) -- (8.6,1)  (9,1) node{...} (6.5,1);
				\draw (0.5,0.9) -- (2.5,0.9);
				
				\draw[anchor=center] (1.4,-0.2) node {$\hat v \in N'$} (2.6,-0.3) node {$v'$} (4,-0.4) node (b1) {$v$} (5.2,-0.4) node {$\beta'$} (3,-0.9) node {$x$};
				\draw[anchor=center] (0,-0.6)  node{$\beta=$}  (0.5,-0.6)  node{$|$} -- (2,-0.6) -- (4.5,-0.6) node{$|$} -- (5.5,-0.6) -- (8.6,-0.6)  (9,-0.6) node{...} (6.5,-0.6);
				\draw (0.5,-0.5) -- (3.5,-0.5);
				\draw (0.5,-0.4) -- (2,-0.4);
				\draw (2,-0.7) -- (4.5,-0.7);
				\draw[shorten <=0.00cm, shorten >=0.04cm,decorate,decoration={snake,amplitude=.4mm,segment length=2mm,post length=1mm},->] (1.4,0.5) to node[left] {$\exists$} (1.4,-0.1);			
				\end{tikzpicture}
			\end{center}
			\caption{Factorization of $\alpha$ and $\beta$ in the proof of \refprop{prp:algebra2monomial}}
		\end{figure}
	
	Thus, let $ux \in A_0^* a_1 A_1^* a_2 \cdots A_{n-1}^*
	a_n A_n^*$ where the monomial has degree at most $k$, 
	then there exists a factorization $A_0^* a_1 A_1^* a_2 \cdots A_{n-1}^*a_n A_n^* = M_1 M_2$ with $M_1, M_2$ $k$-monomials such that $u \in M_1$ and $x \in M_2$. By definition of $N'$ we have $u, \hat v \in N' \subseteq M_1$ and thus $\hat v \in M_1$. 
	We conclude that $v = \hat vx \in M_1 M_2 = A_0^* a_1 A_1^* a_2 \cdots A_{n-1}^*a_n A_n^*$.
	
	Let now $v = \hat v x \in A_0^* a_1 A_1^* a_2 \cdots A_{n-1}^*a_n A_n^*$. 
	Again, there exists a factorization of the monomial $A_0^* a_1 A_1^* a_2 \cdots A_{n-1}^*a_n A_n^* = M_1 M_2$ with 
	$M_1, M_2$ $k$-monomials such that $\hat v \in M_1$ and $x \in M_2$.
	Since $(c_1c_2\cdots c_m)^k$ is a subword of $x$, there must be an $A_i$ in $M_2$ such that $C \subseteq A_i$. Thus, there is a factorisation $M_2 = M_{21} M_{22}$ in
	$k$-monomials $M_{21}, M_{22}$ such that $x' \in M_{21}$, $x'' \in M_{22}$ for $x = x' x''$ and 
	we have $M_{21} \cdot C^* = M_{21}$. 
	Consider $\beta = \hat v x \beta' \in M_1 M_{21} \cdot C^\infty$. Since $\alpha \equiv_\ell^\infty \beta$, we obtain $\alpha \in M_1 M_{21} \cdot C^\infty$. Thus, there is some prefix of $u$ in $M_1 M_{21}$ and by $M_{21} \cdot C^* = M_{21}$, we also obtain $ux' \in M_1 M_{21}$. 
	Thus, $ux = ux' \cdot x'' \in M_1M_{21} \cdot M_{22} = M_1M_2 = A_0^* a_1 A_1^* a_2 \cdots A_{n-1}^*a_n A_n^*$ holds.
	We conclude $ux \equiv_k v$ and thus $t = h(v) = h(u)h(x) \in sh(C^*)$.
\end{proof}

The direct product of homomorphisms $g : \Gamma^* \to M$ and $h : \Gamma^* \to N$ is given by $(g\times h) : \Gamma^* \to M\times N, w \mapsto (g(w),h(w))$. It is well-known, that the direct product recognizes Boolean combinations: 
\begin{lemma}\label{lem:directproductsynt}
	Let $L$ and $K$ be languages such that $L$ recognized by $g : A^* \to M$ and $K$ is recognized by $h: A^* \to N$. Then, any Boolean combination of $L$ and $K$ is recognized by $(g \times h)$.
\end{lemma}
\begin{proof}
	Since $L \cap [s][e]^\omega \neq \emptyset$ implies $[s][e]^\omega \subseteq L$ for some linked pair $(s,e)$, we obtain $\overline L = \cup \set{[s][e]^\omega}{{[s][e]^\omega \cap \overline L \neq \emptyset} }$ for the complement of $L$.
	Thus, it suffices to show that $L \cup K$ is recognized by $(g\times h)$. 
	Obviously, $L$ is covered by $[(s,t)][(e,f)]^\omega$, where $(s,e)$ is a linked pair of $M$ with $[s][e]^\omega \subseteq L$ and $(t,f)$ is any linked pair of $N$. Similiarly one can cover $K$ and thus $M\times N$ recognizes $L \cup K$.
\end{proof}

Next, we show that the algebraic characterisation $\Vtwo$ of $\B\Sigmatwo$ over finite words also holds over finite and infinite words simultaneously.
The proof of this is based on the fact that the algebraic part of the characterisation of $\Sigmatwo$ over finite words and finite and infinite words is the same \cite{dk11tocs}. Since every language of $\Sigmatwo$ is also a language of $\B\Sigmatwo$, and thus $\Vthreehalf \subseteq \Vtwo$, combining this with \reflem{lem:directproductsynt} yields the characterization $\Vtwo$.
\begin{lemma}\label{lem:dd2-B2}
  If $L \subseteq \Gamma^\infty$ is definable in
  $\mathbb{B}\Sigma_2$, then $\Synt(L) \in \Vtwo$.
\end{lemma}

\begin{proof} 
	By definition, $L \in \B\Sigmatwo$ implies that $L$ is a Boolean combination of language $L_i \in \Sigmatwo$. 
	By \cite{dk11tocs} we have $\Synt(L_i) \in \Vthreehalf$ and thus $\Synt(L_i) \in \Vtwo$. 
	Since $L$ is a Boolean combination of $L_i$, $L$ is recognized by the direct product of all $\Synt(L_i)$ by \reflem{lem:directproductsynt}. In particular, $\Synt(L)$ is a divisor of the direct product of $\Synt(L_i)$ by \reflem{lem:syntminimal}. 
	Hence, we obtain $\Synt(L) \in \Vtwo$.
\end{proof}
The proof that monomials are definable in $\Sigmatwo$ is straightforward.
\begin{lemma}\label{lem:mon-sig2}
Let $L\subseteq \Gamma^\infty$ be a monomial of the form $A_0^* a_1 A_1^* a_2 \cdots A_{n-1}^*a_n A_n^\infty$. Then $L$ is definable in 
$\Sigma_2$ by a formula with quantifier depth at most $n + 1$.
\end{lemma}

\begin{proof}
A formula which describes exactly the elements of the monomial is
\begin{align*}
\exists x_1\ldots \exists x_n \forall y& : \bigwedge_{i=1}^n \lambda(x_i) = a_i \land 
\bigwedge_{i=1}^{n-1} x_i < y < x_{i+1} \Rightarrow \lambda(y) \in A_i\land\\&
(y > x_n \Rightarrow \lambda(y) \in A_n) \land (y < x_1 \Rightarrow \lambda(y) \in A_0).
\end{align*}
Hence $L$ is definable in $\Sigma_2$.
\end{proof}
Combining our results we are ready to state and prove the main theorem of the paper.
\begin{theorem}\label{thm:main}
Let $L \subseteq \Gamma^\infty$ be $\omega$-regular. Then the following are equivalent:
\begin{enumerate}
\item $L$ is a finite Boolean combination of monomials of the form $A_0^* a_1 A_1^* a_2 \cdots A_{n-1}^*a_n A_n^\infty$\label{main:aaaa}.
\item $L$ is definable in $\B \Sigma_2$\label{main:bbbb}.
\item The syntactic homomorphism $h$ of $L$ satiesfies:\label{main:cccc}
 \begin{enumerate}
 \item\label{main:cccc:a} $\Synt(L)\in \Vtwo$ and
 \item for all linked pairs $(s,e),(t,f)$ it holds that if there exists an alphabet $C$ and words $\hat e, \hat f$ with 
$h(\hat e) = e, h(\hat f) = f$, $\alph(\hat e) = \alph(\hat f) = C$ and $s\cdot h(C^*) = t\cdot h(C^*)$, 
then $[s][e]^\omega \subseteq L \iff [t][f]^\omega \subseteq L$.\label{main:cccc:b}
 \end{enumerate}
\end{enumerate}
\end{theorem}
\begin{proof}
``\ref{main:aaaa} $\Rightarrow$ \ref{main:bbbb}'':
Since $\B\Sigma_2$ is 
closed under Boolean combinations, it suffices to find a formula in $\Sigma_2$ for the monomials of the form $A_0^* a_1 A_1^* a_2 \cdots A_{n-1}^*a_n A_n^\infty$. 
Hence \reflem{lem:mon-sig2} completes the proof.

``\ref{main:bbbb} $\Rightarrow$ \ref{main:cccc}'': 
\ref{main:cccc:a} is proved by \reflem{lem:dd2-B2}. 
Since $A_0^* a_1 A_1^* a_2 \cdots A_{n-1}^*a_n$ is a set of finite words,  a monomial $A_0^* a_1 A_1^* a_2 \cdots A_{n-1}^*a_n A_n^\infty$ is open in the alphabetic topology by definition. 
 The languages in $\Sigma_2$ are unions of such monomials \cite{dk11tocs} and thus 
 languages in $\B\Sigma_2$ are Boolean combinations of open sets. This implies \ref{main:cccc:b} by \refthm{thm:boolalphopen}.

``\ref{main:cccc} $\Rightarrow$ \ref{main:aaaa}'': 
This is \refprop{prp:algebra2monomial}.
\end{proof}

\begin{example}
In this example we show that the topological property is necessary. For this define $L = (\oneset{a,b}^*aa\oneset{a,b}^*)^\omega$. 
We will show that $\Synt(L) \in \Vtwo$, but $L$ is not a Boolean combinations of open sets of the alphabetic topology. 
Computing the syntactical monoid of $L$ yields $\Synt(L) = \oneset{1,a,b,aa,ab,ba}$. 
The equations  $b^2=b$, $xaa = aax = aa$ and $bab = b$ hold in $\Synt(L)$. 
In particular, $(ab)^2 = ab$ and $(aa)^2 = aa$. Thus, $(s,e) = (aa,aa)$ and $(t,f)=(aa,ab)$ are linked pairs. 
Let $h$ denote the syntactic homomorphism of $L$. 
Choosing $aab$ as a preimage for $aa\in \Synt(L)$ yields the alphabetical condition $\alph(aab) = \alph(ab) = C$ on the idempotents.
 Since $s=t$, we trivially have $s\cdot h(C^*) = t\cdot h(C^*)$. 
However, $[aa][ab]^\omega \cap L = \emptyset$ but $[aa]^\omega \subseteq L$. Thus, $L$ does not satisfy the topological condition.
It remains to check $\Synt(L) \in\Vtwo$. It is enough to show that the preimages are in $\B\Sigmatwo$.
\begin{multicols}{3}
	\begin{itemize}
		\item $[1] = 1$
		\item $[a] = a$
		\item $[b] = b^+ \cup (b^+ab^+)^+$
		\item $[ab] = (ab)^+$
		\item $[ba] = (ba)^+$
		\item $[aa] = \oneset{a,b}^*aa\oneset{a,b}^*$
	\end{itemize}
\end{multicols}
One can find $\B\Sigmatwo$ formulas for these languages, e.g., $[ab] = L(\varphi)$ with \begin{align*}
\varphi \,\equiv\; &(\exists x \forall y\!: x \leq y \land \lambda(x) = a) \land (\exists x \forall y\!: x \geq y \land \lambda(x) = b) \land \\&(\forall x\forall y\!: x\geq y \lor (\exists z: x < z < y) \lor (\lambda(x) \neq \lambda(y))
\end{align*}
and thus $\Synt(L) \in \Vtwo$.
\end{example}

\section{Summary and Open Problems}

The alphabetic topology is an essential ingredient in the study of the fragment $\Sigmatwo$. 
Thus, in order to study Boolean combinations of $\Sigmatwo$ formulas, i.e., the fragment $\B\Sigmatwo$ over infinite words, 
we looked closely at properties of Boolean combinations of its open sets. 
It turns out, that it is decidable whether a regular language is a Boolean combination of open sets. 
This does not follow immediately from the decidability of the open sets.
We used linked pairs of the syntactic homomorphism (which are effectively computable) to get decidability of the topological condition. 
Combining this result with the decidability of $\Vtwo$ we obtained an effective characterization of $\B\Sigmatwo$ over $\Gamma^\infty$, the finite and infinite words over the alphabet~$\Gamma$.

In this paper we dealt with $\B\Sigmatwo$, which is the second level of the Straubing-Th\'erien hierarchy. 
Another well-known hierarchy is the dot-depth hierarchy. 
On the level of logic, the difference between the Straubing-Th\'erien hierarchy and the dot-depth hierarchy is that formulas for the dot-depth hierarchy may also use the successor predicate. 
A deep result of Straubing is that over finite words each level of the Straubing-Th\'erien hierarchy is decidable if and only if it is decidable in the dot-depth hierarchy \cite{str85jpaa}. 
Thus, the decidability result for $\B\Sigmatwo$ by Place and Zeitoun also yields a decidability result of $\B\Sigmatwo[<,{+}1]$. 
The fragment $\Sigmatwo[<,{+}1]$ is decidable for $\omega$-regular languages \cite{kkl11stacs:short}. 
This result also uses topological ideas, namely the factor topology. The open sets in this topology describe which factors of a certain length $k$ may appear in the ``infinite part'' of the words. 
The study of Boolean combinations of open sets in the factor topology is an interesting line of future work, and it may yield a decidability result for $\B\Sigmatwo[<,{+}1]$ over infinite words.

Another interesting class of predicates are modular predicates. In \cite{KufleitnerWalter2015rairo} the authors have studied $\Sigmatwo[<,\mathrm{MOD}]$ over finite words. 
The results of \cite{KufleitnerWalter2015rairo} can be generalised to infinite words by adapting the alphabetic topology to the modular setting. As for successor predicates, we believe that an appropriate effective characterization of this topology might help in deciding 
$\B\Sigmatwo[<,\mathrm{MOD}]$ over infinite words. To the best of our knowledge however, modular predicates have not yet been considered over infinite words.

\bibliographystyle{plain}

\begin{thebibliography}{10}
	
	\bibitem{bk78jcss:short}
	Janusz~Antoni Brzozowski and Robert Knast.
	\newblock The dot-depth hierarchy of star-free languages is infinite.
	\newblock {\em J.\ Comput.\ Syst.\ Sci.}, 16(1):37--55, 1978.
	
	\bibitem{dg08SIWT:short}
	Volker Diekert and Paul Gastin.
	\newblock First-order definable languages.
	\newblock In {\em Logic and Automata: History and Perspectives}, Texts in Logic
	and Games, pages 261--306. Amsterdam University Press, 2008.
	
	\bibitem{dgk08ijfcs:short}
	Volker Diekert, Paul Gastin, and Manfred Kufleitner.
	\newblock A survey on small fragments of first-order logic over finite words.
	\newblock {\em Int.\ J.\ Found.\ Comput.\ Sci.}, 19(3):513--548, 2008.
	
	\bibitem{dk11tocs}
	Volker Diekert and Manfred Kufleitner.
	\newblock Fragments of first-order logic over infinite words.
	\newblock {\em Theory of Computing Systems}, 48(3):486--516, 2011.
	
	\bibitem{eil76:short}
	Samuel Eilenberg.
	\newblock {\em Automata, Languages, and Machines}, volume~B.
	\newblock Academic Press, 1976.
	
	\bibitem{kkl11stacs:short}
	Jakub Kallas, Manfred Kufleitner, and Alexander Lauser.
	\newblock First-order fragments with successor over infinite words.
	\newblock In {\em STACS 2011, Proceedings}, volume~9 of {\em LIPIcs}, pages
	356--367. Dagstuhl Publishing, 2011.
	
	\bibitem{KufleitnerWalter2015rairo}
	Manfred Kufleitner and Tobias Walter.
	\newblock One quantifier alternation in first-order logic with modular
	predicates.
	\newblock {\em RAIRO-Theor. Inf. Appl.}, 49(1):1--22, 2015.
	
	\bibitem{Landweber69}
	Lawrence~H. Landweber.
	\newblock Decision problems for $\omega$-automata.
	\newblock {\em Mathematical Systems Theory}, 3(4):376--384, 1969.
	
	\bibitem{mcnaughton66}
	Robert McNaughton.
	\newblock Testing and generating infinite sequences by a finite automaton.
	\newblock {\em Information and Control}, 9:521--530, 1966.
	
	\bibitem{pp04:short}
	Dominique Perrin and Jean-{\'E}ric Pin.
	\newblock {\em Infinite words}, volume 141 of {\em Pure and Applied
		Mathematics}.
	\newblock Elsevier, 2004.
	
	\bibitem{pin95:short}
	Jean-{\'E}ric Pin.
	\newblock A variety theorem without complementation.
	\newblock In {\em Russian Mathematics (Iz.\ VUZ)}, volume~39, pages 80--90,
	1995.
	
	\bibitem{pin97handbook:short}
	Jean-{\'E}ric Pin.
	\newblock Syntactic semigroups.
	\newblock In {\em Handbook of Formal Languages}, volume~1, pages 679--746.
	Springer, 1997.
	
	\bibitem{PlaceZeitoun14icalp}
	Thomas Place and Marc Zeitoun.
	\newblock Going higher in the first-order quantifier alternation hierarchy on
	words.
	\newblock In Javier Esparza, Pierre Fraigniaud, Thore Husfeldt, and Elias
	Koutsoupias, editors, {\em Automata, Languages, and Programming - 41st
		International Colloquium, {ICALP} 2014, Copenhagen, Denmark, July 8-11, 2014,
		Proceedings, Part {II}}, volume 8573 of {\em Lecture Notes in Computer
		Science}, pages 342--353. Springer, 2014.
	
	\bibitem{ram30:short}
	Frank~Plumpton Ramsey.
	\newblock On a problem of formal logic.
	\newblock {\em Proc.\ London Math.\ Soc.}, 30:264--286, 1930.
	
	\bibitem{SchwarzStaiger10}
	Sibylle Schwarz and Ludwig Staiger.
	\newblock Topologies refining the cantor topology on
	\emph{X}\({}^{\mbox{\emph{omega}}}\).
	\newblock In Cristian~S. Calude and Vladimiro Sassone, editors, {\em
		Theoretical Computer Science - 6th {IFIP} {TC} 1/WG 2.2 International
		Conference, {TCS} 2010, Held as Part of {WCC} 2010, Brisbane, Australia,
		September 20-23, 2010. Proceedings}, volume 323 of {\em {IFIP} Advances in
		Information and Communication Technology}, pages 271--285. Springer, 2010.
	
	\bibitem{sw74eik:short}
	Ludwig Staiger and Klaus~W. Wagner.
	\newblock Automatentheoretische und automatenfreie {C}harakterisierungen
	topologischer {K}lassen regul{\"a}rer {F}olgenmengen.
	\newblock {\em Elektron.\ Inform.-verarb.\ Kybernetik}, 10:379--392, 1974.
	
	\bibitem{str85jpaa}
	Howard Straubing.
	\newblock Finite semigroup varieties of the form {$\mathbf{V}\ast \mathbf{D}$}.
	\newblock {\em Journal of Pure and Applied Algebra}, 36(1):53--94, 1985.
	
	\bibitem{str94:short}
	Howard Straubing.
	\newblock {\em Finite Automata, Formal Logic, and Circuit Complexity}.
	\newblock Birkh{\"a}user, 1994.
	
	\bibitem{tho82:short}
	Wolfgang Thomas.
	\newblock Classifying regular events in symbolic logic.
	\newblock {\em J.\ Comput.\ Syst.\ Sci.}, 25:360--376, 1982.
	
	\bibitem{tho90handbook:short}
	Wolfgang Thomas.
	\newblock Automata on infinite objects.
	\newblock In {\em Handbook of Theoretical Computer Science}, chapter~4, pages
	133--191. Elsevier, 1990.
	
	\bibitem{wilke93stacs}
	Thomas Wilke.
	\newblock Locally threshold testable languages of infinite words.
	\newblock In {\em STACS '93, Proceedings}, volume 665 of {\em LNCS}, pages
	607--616. Springer, 1993.
	
\end{thebibliography}
\newcommand{\Ju}{Ju}\newcommand{\Ph}{Ph}\newcommand{\Th}{Th}\newcommand{\Ch}{Ch}\newcommand{\Yu}{Yu}\newcommand{\Zh}{Zh}\newcommand{\St}{St}\newcommand{\curlybraces}[1]{\{#1\}}

\end{document}